\documentclass[journal,final,twocolumn]{IEEEtran}
\usepackage[T1]{fontenc}
\usepackage{graphicx}
\usepackage{cite}
\usepackage{amsmath,amssymb,amsfonts}
\usepackage{algorithmic}
\usepackage{graphicx}
\usepackage{textcomp}
\usepackage{xcolor}
\usepackage{amsthm}
\usepackage{tikz}
\usepackage{color,soul}
\usetikzlibrary{svg.path}
\usepackage{scalerel}
\usepackage{subcaption}
\usepackage{floatrow}
\usepackage{dutchcal}
\usepackage{mathrsfs}
\usepackage[nolist,nohyperlinks]{acronym}
\makeatletter
\expandafter\patchcmd\csname AC@\AC@prefix{}@acro\endcsname{{#3}}{{\MakeUppercase #3}}{}{}
\expandafter\patchcmd\csname AC@\AC@prefix{}@acro\endcsname{{#3}}{{\MakeUppercase #3}}{}{}
\makeatother

\begin{acronym}
    \acro{BPP}[BPP]{binomial point process}
    \acro{BS}[BS]{base station}
    \acro{CDF}[CDF]{cumulative distribution function}
    \acro{DL}[DL]{downlink}
    \acro{eMBB}[eMBB]{enhanced mobile broadband}
    \acro{DSA}[DSA]{dynamic spectrum access}
    \acro{GEO}[GEO]{geostationary earth orbit}
    \acro{HAPS}[HAPS]{high-altitude platform system}
    \acro{ITU}[ITU]{International Telecommunication Union}
    \acro{KPI}[KPI]{key performance indicator}
    \acro{LEO}[LEO]{low Earth orbit}
    \acro{NBIOT}[NB-IoT]{narrowband-IoT}
    \acro{NTN}[NTN]{\protect\FirstLetter{N}{n}on-terrestrial networks}
    \acro{ORAN}[ORAN]{open radio-access network}
    \acro{SINR}[SINR]{signal-to-interference-plus-noise ratio}
    \acro{SR}[SR]{shadowed-Ricean}
    \acro{TN}[TN]{terrestrial networks}
    \acro{UAV}[UAV]{unmanned aerial vehicle}
    \acro{UE}[UE]{user equipment}
    \acro{UL}[UL]{uplink}
    \acro{ISD}[ISD]{inter-site distance}
    \acro{6G}[6G]{sixth-generation}
    \acro{5G}[5G]{fifth-generation}
\end{acronym}

\newcommand{\SINRTN}{\mathrm{SINR}_{\mathrm{}}}

\newcommand{\rmax}{r_{\mathrm{max}}}
\newcommand{\riso}{d_{\mathrm{iso}}}
\newcommand{\rEarth}{r_{\oplus}}

\newcommand{\isd}{d_{\mathrm{ISD}}}
\newcommand{\Prob}{{\mathbb{P}}}

\newcommand{\Ptn}{p_\mathrm{TN}}
\newcommand{\Gtn}{G_\mathrm{TN}}
\newcommand{\Htn}{H_\mathrm{TN}}

\newcommand{\Itn}{I_\mathrm{TN}}
\newcommand{\Nc}{N_\mathrm{c}}
\newcommand{\Nu}{N_\mathrm{u}}

\newcommand{\Pntn}{p_{\mathrm{NTN}}}
\newcommand{\Gntn}{G_\mathrm{NTN}}
\newcommand{\Hntn}{H_\mathrm{NTN}}
\newcommand{\Intn}{I_\mathrm{NTN}}
\newcommand{\IntnI}{I_{\mathrm{DL}}}
\newcommand{\IntnII}{I_{\mathrm{UL}}}
\newcommand{\Pc}{P_\mathrm{c}}

\newcommand{\pt}{p_\mathrm{t}}

\newcommand{\Rs}{R_\mathrm{0}}
\newcommand{\Rntn}{{\mathcal{R}}_n}
\newcommand{\Rntnsat}{{\mathscr{R}}_n}
\newcommand{\rntnsat}{{\mathcal{r}}_n}

\newcommand{\rs}{r_\mathrm{0}}
\newcommand{\rntn}{r_\mathrm{NTN}}

\newcommand{\rn}{r_n}
\newcommand{\Rn}{R_n}
\newcommand{\rTN}{r_\mathrm{TN}}

\newcommand{\Sc}{s_\mathrm{c}}
\newcommand{\sr}{s_\mathrm{r}}

\newcommand{\atn}{\alpha_{\mathrm{TN}}}
\newcommand{\antn}{\alpha_{\mathrm{NTN}}}

\newtheorem{theorem}{Theorem}
\newtheorem{lem}{Lemma}

\definecolor{orcidlogocol}{HTML}{A6CE39}
\tikzset{
  orcidlogo/.pic={
    \fill[orcidlogocol] svg{M256,128c0,70.7-57.3,128-128,128C57.3,256,0,198.7,0,128C0,57.3,57.3,0,128,0C198.7,0,256,57.3,256,128z};
    \fill[white] svg{M86.3,186.2H70.9V79.1h15.4v48.4V186.2z}
                 svg{M108.9,79.1h41.6c39.6,0,57,28.3,57,53.6c0,27.5-21.5,53.6-56.8,53.6h-41.8V79.1z M124.3,172.4h24.5c34.9,0,42.9-26.5,42.9-39.7c0-21.5-13.7-39.7-43.7-39.7h-23.7V172.4z}
                 svg{M88.7,56.8c0,5.5-4.5,10.1-10.1,10.1c-5.6,0-10.1-4.6-10.1-10.1c0-5.6,4.5-10.1,10.1-10.1C84.2,46.7,88.7,51.3,88.7,56.8z};
  }
}

\newcommand\orcidicon[1]{\href{https://orcid.org/#1}{\mbox{\scalerel*{
\begin{tikzpicture}[yscale=-1,transform shape]
\pic{orcidlogo};
\end{tikzpicture}
}{|}}}}

\usepackage{hyperref} 

\title{Co-existence of Terrestrial and \\
Non-Terrestrial Networks in S-band}
\author{\IEEEauthorblockN{Niloofar Okati \orcidicon{0000-0002-8074-5146}, Andre Noll Barreto, Luis Uzeda Garcia \orcidicon{0000-0002-0947-6852}, and Jeroen Wigard}%
}
\begin{document}
\newcommand*\FirstLetter[2]{#1}
\renewenvironment{table}%
  {\renewcommand\familydefault\sfdefault
   \@float{table}}
  {\end@float}
\makeatother

\setlength{\parskip}{0pt}


\maketitle

\begin{abstract}
Co-existence of terrestrial and non-terrestrial networks (NTN) is foreseen as an important component to fulfill the global coverage promised for sixth-generation (6G) of cellular networks. Due to ever rising spectrum demand, using dedicated frequency bands for terrestrial network (TN) and NTN may not be feasible. As a result, certain S-band frequency bands allocated by radio regulations to NTN networks are overlapping with those already utilized by cellular TN, leading to significant performance degradation due to the potential co-channel interference. Early simulation-based studies on different co-existence scenarios failed to offer a comprehensive and insightful understanding of these networks’ overall performance. Besides, the complexity of a brute force performance evaluation increases exponentially with the number of nodes and their possible combinations in the network. In this paper, we utilize stochastic geometry to analytically derive the performance of TN-NTN integrated networks in terms of the probability of coverage and average achievable data rate for two co-existence scenarios. From the numerical results, it can be observed that, depending on the network parameters, TN and NTN users’ distributions, and traffic load, one co-existence case may outperform the other, resulting in optimal performance of the integrated network. The analytical results presented herein pave the way for designing state-of-the-art methods for spectrum sharing between TN and NTN and optimizing the integrated network performance.
\end{abstract}


\IEEEpeerreviewmaketitle
\section{Introduction}
\label{sec:intro}
\ac{NTN} is foreseen as an integral part of sixth-generation (6G) cellular networks to fulfill the requirements on expanding coverage and connectivity in remote and underserved areas \cite{LEO2021Xie}. \ac{NTN} encompass different types of networks communicating through the sky (air or space), including satellite-based networks, \acp{HAPS}, and \acp{UAV}. During the recent years, due to the significant reduction in the launch costs of satellites and increased demand for global broadband connectivity, \ac{LEO} satellites outpaced other types of \ac{NTN} in commercialization, leading to several initiatives such as Starlink, OneWeb, AST SpaceMobile and Project Kuiper \cite{lin2021}.

The integration of \ac{NTN} technology into \ac{TN} for the next-generation communication systems, such as 6G, holds great potential for revolutionizing global connectivity \cite{azari2022}, by complementing terrestrial infrastructure and extending connectivity to remote areas, such as rural regions, maritime environments, and disaster-stricken locations. However, the successful deployment and operation of \ac{NTN} systems face several challenges, one of which is the scarcity of spectrum resources. The limited availability of suitable frequency bands may impose that \ac{NTN} and \ac{TN} systems operate in shared frequency bands, which will inevitably cause interference between both systems. To address this limitation, it becomes imperative to explore mechanisms for sharing and co-existence between \ac{NTN} and \ac{TN}, while ensuring efficient spectrum utilization and minimizing interference.

 Due to the massive number of nodes and the large number of possible combinations, Monte Carlo simulation-based performance evaluation of the TN-NTN integrated network is computationally demanding. In this paper, to evaluate the interference effect of an \ac{NTN} network on the performance of a \ac{TN} network, we utilized stochastic geometry as a powerful mathematical framework which has a rich history in modeling and analyzing of wireless networks \cite{baccelli2009stochastic}. Stochastic geometry enables obtaining analytical tractable derivations on the performance of the integrated network using the spatial distribution of network elements, e.g., \acp{BS} and \acp{UE}. 
 The analysis provides several insights into the impact of different network parameters, such as the density of TN \acp{BS}, NTN \acp{BS}' altitudes, \ac{ISD}, the number of UEs, and the minimum separation distance between the users of each network, on the network's performance. Obviously, such results will pave the way for design of different co-existence mechanisms and spectrum sharing policies required to enable seamless operation of \ac{NTN} and \ac{TN} systems.



\subsection{Related Works}
The integration of \ac{NTN} and \ac{TN} is proposed and investigated in many prior arts in the literature. In \cite{lopez2022},
the performance of \ac{5G} \ac{TN} is compared with that of a \ac{LEO} satellite network using an experimental setup in a suburban environment, in terms of latency and throughput. It was demonstrated that multi-connectivity between both networks can provide coverage seamlessly for low-latency and high-reliability requiring services.
Co-existence and spectrum sharing studies have received significant attention in the \ac{ITU} agenda. For instance, a simulation tool that analyses the aggregate interference between \ac{5G} \ac{TN} and other networks following ITU recommendations was presented in \cite{souza2017}. It was used, for instance, in \cite{queiroz2019} to investigate co-existence of 5G and \ac{HAPS}. 

Several methods for spectrum sharing between \ac{TN} and \ac{LEO} satellites have been suggested in the literature. In  \cite{smith2021}, a methodology  using \ac{ORAN} was proposed, in which interference is reported by the ORAN terrestrial networks to the satellite service provider. In \cite{lee2021}, a simulation-based study for spectrum allocation between TN and NTN was developed to mitigate the co-channel interference to the NTN link. An experimental-based study on spectrum sharing between TN and a LEO satellite network was presented in \cite{Kokkinen2023Proof}, targeting as minimizing the interference from satellites at the TN users.  Cognitive radio approaches for co-existence of satellites and TN cellular networks over the same frequency bands have been studied in \cite{Resource2015Lagunas,Genetic2016Icolari}, targeting at better exploitation of frequencies that are not used by the terrestrial network.  
An overview of the integration of both \ac{GEO} and \ac{LEO} satellites with 5G systems, focusing on \ac{eMBB} and \ac{NBIOT} was presented in \cite{charbit2021}. 


Stochastic geometry is a mathematical tool that finds extensive application across various scientific disciplines including wireless communication. Its utilization in the analysis of wireless networks dates back to 1961 to investigate the connectivity of large wireless networks {\color{black}with huge number of nodes} \cite{Gilber1961}.  
Stochastic geometry is utilized for system level performance analysis of large wireless networks to avoid massive time consuming Monte Carlo simulations. Using stochastic geometry, the performance of the network is obtained by taking the averages over different realization of nodes in a network \cite{baccelli2009stochastic}. Averages can be computed over a large number of nodes' locations or several network realizations, for different metrics of interest, e.g., the probability of coverage and the average achievable rate. 
More particularly, the nodes' locations are modeled as a point process which captures the characteristics of their distribution, e.g., their density and the relative distances between them. 

Stochastic geometry has found extensive application in the analysis of planar two-dimensional networks including heterogeneous \cite{Cao2013Optimal,Dhillon2012modeling,Dhillon2013LoadAware,heterogeneousSurvey,heterogen2016Wei}, cognitive \cite{Lee2012Interference,Ekram2015Congnitive}, ad hoc \cite{Hasan2007adhoc,Hanter2008adhoc,ganti2009adhoc,Baccelli2009adhoc}, and vehicular \cite{Chetlur2020vehicular,jarajaj2021vehicular} networks. A comprehensive literature survey on utilization of stochastic geometry for modeling multi-tier and cognitive cellular networks has been provided in \cite{heterogeneousSurvey}. A taxonomy is founded on three key criteria: the type of network under consideration, the point process employed for modeling, and the network performance characterization techniques. The coverage and rate are tractably derived through modeling the \acp{BS} as a homogeneous PPP in \cite{Tractable2011Andrews} which is known as one of the seminal works on the utilization of stochastic geometry for performance analysis of a terrestrial cellular. The model shows the optimistic results of the conventional grid models while the stochastic model provides a lower band on the network performance. 

The utilization of stochastic geometry for analyzing three-dimensional networks has been also compelling \cite{PAN2015Modeling,Chetlur2017DownlinkUAV,Wang2018UAV}. The tool was used to study different scenarios such as \acp{BS} being located on both ground and rooftops in highly crowded urban areas, UAV swarming, and massive satellite networks. Two spectrum sharing techniques, i.e., underlay and overlay, to share the spectrum between UAV-to-UAV and \ac{BS}-to-UAV transceivers are analyzed in \cite{Azari2020UAV} by modeling both UAVs and \acp{BS} as homogeneous PPPs. The impact of interference from TN \acp{BS} on the primary users of a multi-beam geostationary satellite when both networks share the same spectrum resources is investigated in \cite{kolawole2017Performance}. Using the tools from stochastic geometry, authors analyzed two performance metrics of outage probability and area spectral efficiency under three secondary transmission schemes by modeling the satellite users and \acp{BS} as two independent point processes.

The application of stochastic geometry has been also extended for performance analysis of massive low Earth orbit 
constellations, by modeling the satellites as a \ac{BPP} \cite{okati2020downlink,okati2020PIMRC,talgat2020stochastic,Alouini2020NearestNeighbor} or a nonhomogeneous Poisson point process (NPPP) \cite{okatiNonhomo,okatiVTC2021} on a spherical shell which facilitates the utilization of the tools from stochastic geometry. In fact, a point process represents different realizations of satellites over the time caused by the motion of satellites on the orbits. Verification of the exact analytical derivations from the stochastic modeling with the actual simulated constellations represents a fair accuracy of such modeling. In \cite{okatiWCNC2022}, the coverage and rate of a noise-limited scenario is studied using stochastic geometry when the user is associated with a satellite which provides the highest signal-to-noise (SNR). Such association, compared to the simplistic case of associating the user with the nearest satellites, includes the effect of shadowing caused by the user's surrounding objects. The performance of multi-altitude LEO constellations when satellites are distributed on several orbital shells is studied in \cite{talgat2020stochastic, okatiLetter1}.

In this paper, we study the coverage probability and average data rate of a primary TN network which shares the same spectrum resources with a secondary LEO network in S-band. We utilize the tools from stochastic geometry to investigate two scenarios in which the TN \ac{DL} transmission is subject to interference from either satellites or NTN UEs transmitting in \ac{UL}. Using the tools from stochastic geometry facilitates derivations of analytical exact expressions for the performance of the integrated TN-NTN network. The results herein provide direct insights on several parameters associated with the integrated network such as the satellite altitude, the TN user location, transmission power of \acp{BS} and satellites, traffic load, and ISD. 

\subsection{Contributions and Paper Organization}
In this paper, we consider a group of TN \acp{BS} distributed over a finite region which is also under the coverage of a beam from a LEO satellite. The TN and LEO networks share the same frequency channels in S-band. The UEs, located further away from the TN network's coverage area, lack access to the service provided by the TN network and, thus, are served by the LEO network. In order to analyze the performance of the integrated network and utilize the tools from stochastic geometry, we model the TN \acp{BS} and the NTN UEs as a \ac{BPP}. The choice of \ac{BPP} is justifiable as i) \acp{BS} and/or NTN users are distributed over a finite region, and ii) the performance is location dependent \cite{afshang2017Fundementals}. 


\begin{figure*}
         \begin{subfigure}[b]{0.48\textwidth}
         \centering
         \includegraphics[trim = 5mm 0mm 5mm 10mm, clip,width=\textwidth]{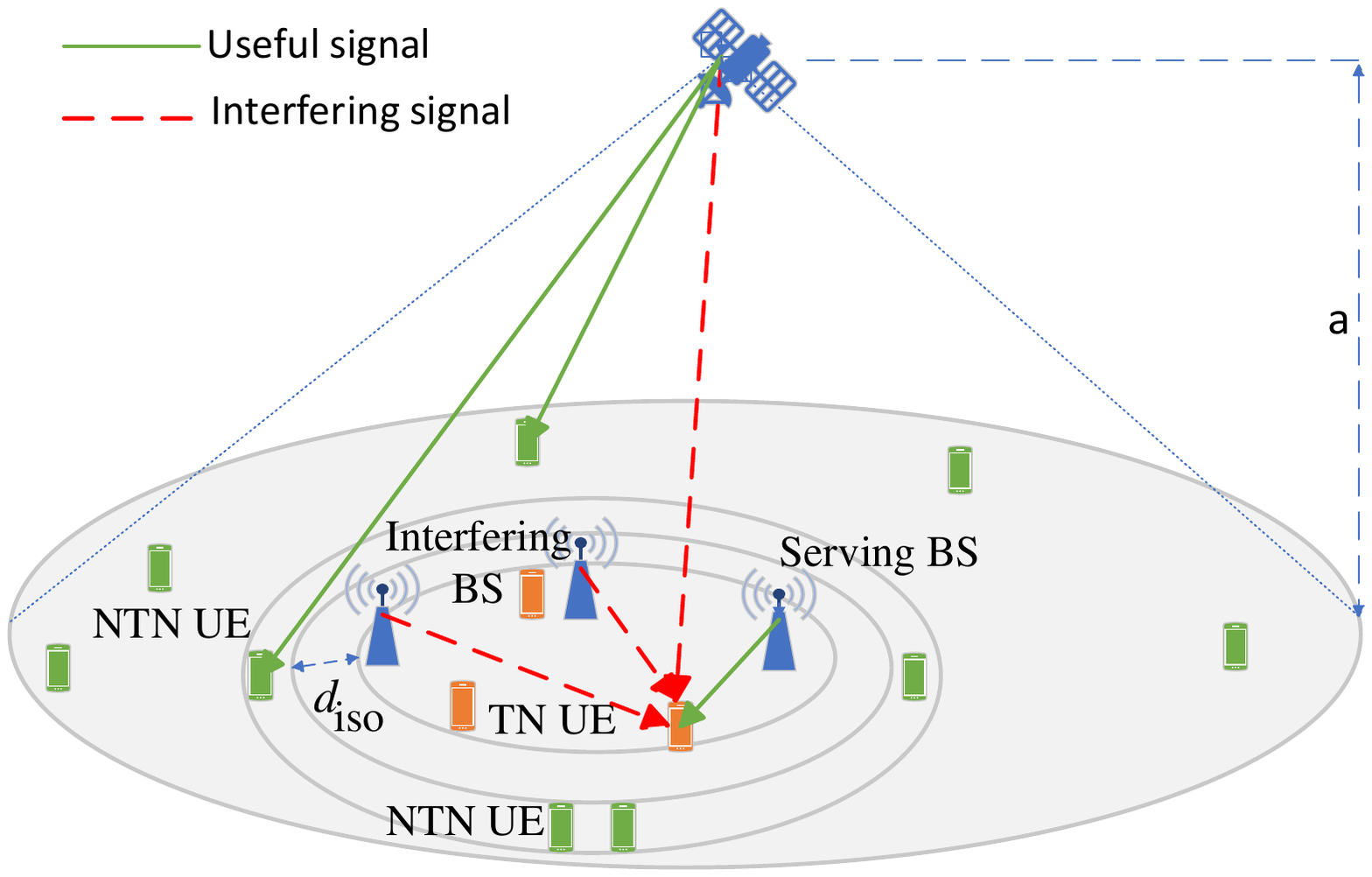}
         \caption{NTN DL interferes with TN DL.}
         \label{fig:sys_scenario3}
     \end{subfigure}
     \hfill
     \begin{subfigure}[b]{0.48\textwidth}
         \centering
         \includegraphics[trim = 5mm 0mm 5mm 10mm, clip,width=\textwidth]{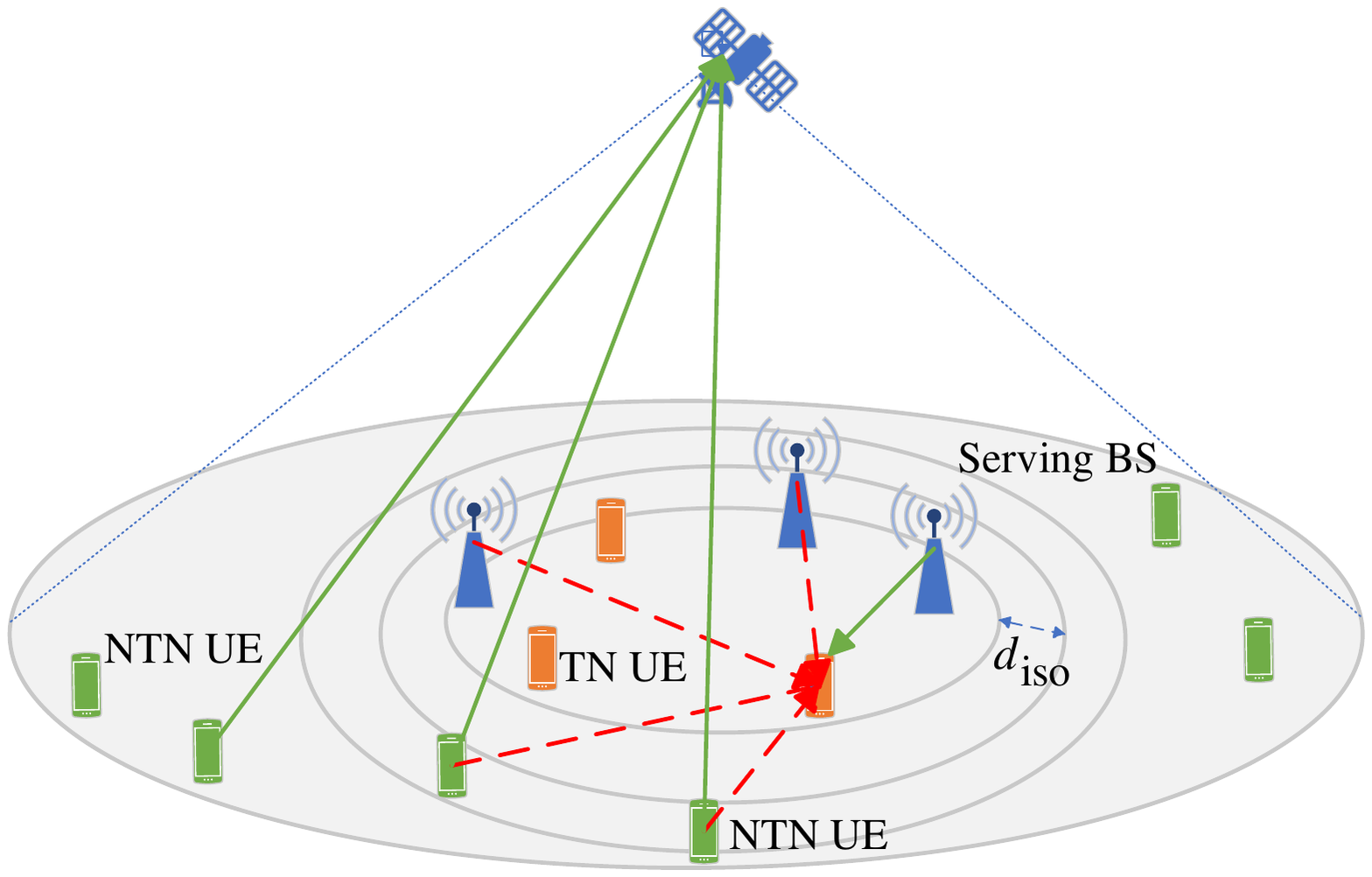}
         \caption{NTN UL interferes with TN DL.}
         \label{fig:sys_scenario5}
     \end{subfigure}
     \hfill 
    \caption{Co-existence scenarios from \cite{TR38.863} where NTN DL (a) or NTN UL (b) interferes with TN DL.}
   \label{fig:system model}
\end{figure*}
The main scientific contributions of this paper are as follows:
\begin{itemize}
	
\item 
Using the tools from stochastic geometry, we derive exact analytical expressions for the probability of coverage and average achievable data rate of a TN \ac{DL} network which shares the same frequency channel with NTN network.

\item 
The performance analysis is provided for two co-existence cases which will be referred to as Case~I and Case~II in what follows. Case~I corresponds to the case when TN network shares the same frequency channel with NTN in \ac{DL} direction, i.e., the source of inter-network interference is from the satellite broadcasting the signal to the area where TN network is located. In Case~II, the TN network shares the same frequency band with NTN in \ac{UL} direction, i.e., the inter-network source of interference is from the NTN UEs transmitting to their serving satellite in close proximity of the TN network. {\color{black}Note that according to \cite{mediatek}, the performance degradation in case of co-existence of TN UL with NTN network is insignificant.}

\item 
We validate all our theoretical derivations with Monte Carlo simulations. Throughout the numerical results, we show how each of the co-existence cases mentioned in the previous bullet point may affect the performance of TN network. The performance is illustrated in terms of several parameters, e.g., TN user's location in the cell, type of geographical environments, satellite's altitude, TN \ac{BS}'s transmit power, the number of NTN users, TN traffic load, and the isolation distance between TN and NTN UEs.
 
\item 
The numerical results suggest several insightful guidelines on spectrum sharing between TN and NTN networks. In other words, the results provide the answer to the following question: {\em{Given the network deployment parameters, user's distributions, and traffic load, which case of co-existence leads to a superior performance for the TN network?}} 
\end{itemize}

For the propagation model, we consider Nakagami-$m$ fading with integer $m$ which not only preserves the analytical tractability but also covers a wide range of fading cases by varying the parameter $m$. As increasing $m$ corresponds to higher probability of line-of-sight, we opt a higher $m$ for satellite to ground channels while a smaller value is used for terrestrial channels. From the numerical results, we will see that there is less performance degradation in urban regions when TN DL shares the same band with NTN DL. In the other scenario when NTN UL uses the same band as TN DL, the performance varies significantly with the TN user's location w.r.t. the cell center, i.e., more degradation occurs for edge users. 

The remainder of this paper is organized as follows. Section~\ref{sec:sysmod} describes the system model as well as the two TN-NTN co-existence scenarios studied in this paper. Then in Section~\ref{sec:cov}, we derive the exact analytical expressions for the coverage performance and average data rate for each co-existence case. The verification of the analytical results as well as studying the effect of several network parameters on the performance of the integrated network is provided in Section~\ref{sec:Numerical Results}. Finally, the paper is concluded in Section~\ref{sec:Conclusion}.
\begin{figure} [h!]
         \centering
         \includegraphics[trim = 10mm 5mm 10mm 5mm, clip,width=\textwidth]{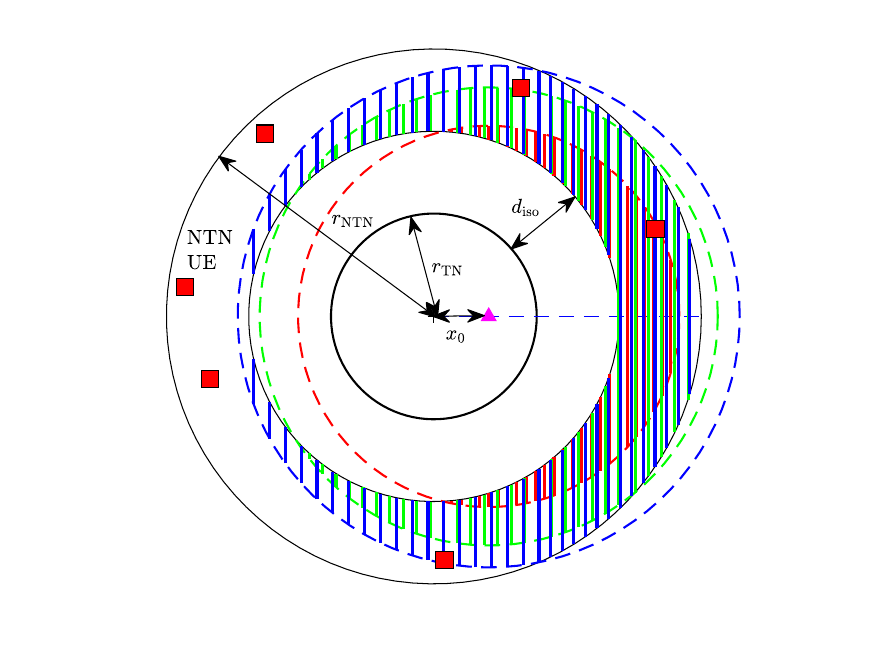}
         \caption{Geometry of the system model used for derivation of \eqref{eq:CDF_Rn_ntn} in co-existence Case~II. The shaded area is the area obtained from the intersection of the disc centered at TN UE with radius $\Rntn$ and the outer annulus defining the outer border for NTN UEs.}
         \label{fig:Integral bounderies}
\end{figure}
\section{System Model}
\label{sec:sysmod}
The \ac{TN}-\ac{NTN} co-existence cases studied in this paper are shown in Fig.~\ref{fig:system model}. In these cases, a \ac{LEO} satellite orbiting at altitude $a$ serves $\Nu$ \ac{NTN} users in \ac{DL} (c.f. Fig.~\ref{fig:system model}(a)) and \ac{UL} (c.f. Fig.~ \ref{fig:system model}(b)). Assuming that the \ac{TN} network is always preferable over \ac{NTN} due to its superior performance and reliability, \ac{NTN} 
\acp{UE} are assumed to be distant from the TN network's edge at {\color{black}least by a distance denoted by $\riso$.} The \ac{TN} \acp{BS} are distributed as clusters with radius $\rTN$, within the beam coverage area of the satellite, with each cluster having $\Nc$ \acp{BS} and the minimum distance between any two \acp{BS}, known as inter-site distance, is denoted by $\isd$, which is a function of the physical environment, e.g., rural, urban, etc. 
\ac{TN} \acp{UE} are uniformly distributed within each \ac{TN} cluster.

In this paper, the serving \ac{TN} \ac{BS} is assumed to be the nearest \ac{BS} to the \ac{TN} user and serves the \ac{TN} \ac{UE} in the DL direction. 
We assume that the same frequency band is shared among all \ac{TN} \acp{BS} 
and with the \ac{NTN} that provides coverage over the \ac{TN} cluster. This causes co-channel interference at the \ac{TN} \ac{UE}'s reception. In Fig.\ref{fig:system model}, solid lines represent the useful signals while dashed lines depict the interfering signals.

\begin{table*}
\renewcommand \caption [2][]{}
\centering
\begin{minipage}{0.9\textwidth}
\setcounter{equation}{2}
\begin{align}
    \label{eq:CDF_R0}
    &F_{R_0}\left(r_0\right)\triangleq\Prob\left(R_0 \leq r_\text{0}\right) = 1 - \left(1-F_{R}\left(\rs\right)\right)^{\Nc}\\\nonumber
&=\left\{
\begin{array}{ll}
1-\left(1-\left(\frac{\rs}{\rTN}\right)^2\right)^{\Nc}, & 0\leq\rs\leq \rTN-x_0\\
 1-\left(1-\frac{1}{\pi}\left(\left(\frac{\rs}{\rTN}\right)^2\left(\theta-\frac{1}{2}\sin(2\theta)\right)+\left(\phi-\frac{1}{2}\sin(2\phi)\right)\right)\right)^{\Nc}, & \rTN-x_0<\rs \leq \rTN+x_0,\\
0,&\text{otherwise}.
\end{array} \right.
\end{align}
\end{minipage}
\end{table*}

The channel model for serving link follows a Nakagami-$m$ fading model which allows to account for a wide range of multi-path fading conditions by varying the parameter $m$, while preserving the analytical tractability. Following the same approach as in \cite{okati2020downlink,okatiNonhomo}, we do not need to limit the fading distribution of the interfering channels to any specific one for the analytical derivations, since it has no effect on the tractability of our analysis. Obviously, resorting to some specific channel models is required to generate the numerical results. As the $m$ parameter represents the ratio between the received power of line-of-sight and non-line-of-sight components, a larger value should be opted for satellite to ground channel compared to terrestrial links with less line-of-sight probability. The channel gains for TN and NTN links are denoted by $\Htn$ and $\Hntn$, respectively.

The antennas at the satellites and \ac{TN} \acp{BS} are directional and the beamforming patterns are chosen according to specifications given in \cite{TR38.863}. The antennas can radiate multiple beams towards the ground. The \ac{TN} and \ac{NTN} users are equipped with omni-directional antenna elements. Assuming perfect beamforming, the overall maximum antenna gain will be expressed as $\Gntn=G_{\mathrm{sat}}G_{\mathrm{u}}$ and $\Gtn=G_{\mathrm{BS}}G_{\mathrm{u}}$, where $G_{\mathrm{sat}}$, $G_{\mathrm{BS}}$, and $G_{\mathrm{u}}$ are the maximum antenna gain of satellites, the \ac{TN} \acp{BS}, and the user, respectively.

We study the \ac{TN}-\ac{NTN} co-existence under {\color{black}{two}} different scenarios as described in the following subsections. In these scenarios, we study the performance degradation caused by spectrum sharing between \ac{TN} and \ac{NTN} networks at the \ac{TN} \ac{UE} during \ac{NTN} network transmission in DL and UL. Based on the described scenarios, the \ac{SINR} at the \ac{TN} \ac{UE} can be expressed as
\setcounter{equation}{0}
\begin{align}
\label{eq:SINR}
\SINRTN = \frac{\Ptn \Gtn \Htn R_0^{-\atn}}{\Itn+\Intn+\sigma^2},
\end{align}
where constant $\sigma^2$ is the power of additive thermal noise, the parameter $\atn$ is a path loss exponent for TN link, $\Ptn$ is the transmit power of TN BS, $\Itn = \sum_{n=1}^{\Nc-1} \Ptn \Gtn \Htn R_n^{-\atn}$
is the cumulative interference power from all other \acp{BS} in the same cluster as the \ac{TN} user, and $\Intn\in\{\IntnI,\IntnII\}$ is the interference power received from the \ac{NTN} network in DL and UL directions denoted by $\IntnI$ and $\IntnII$, respectively. The distances from the TN user to the serving TN BS and other interfering BSs are denoted by random variables $R_0$ and $R_n$, \mbox{$n=1,2,\ldots,\Nc-1$}, respectively. In the following subsections, we elaborate on two TN-NTN co-existence cases studied in this paper.

\subsection{Co-existence Case~I: NTN DL as the aggressor and TN DL as the victim}
\label{subsec:case1}
Co-existence Case~I corresponds to Scenario~3 given in \cite{TR38.863}, in which both \ac{TN} and \ac{NTN} operate in DL over the same frequency band (c.f. Fig.\ref{fig:sys_scenario3}). The aggressor in this case is a \ac{LEO} satellite at altitude $a$ and the victim is a \ac{TN} user, which receives the useful signal from a \ac{BS}. The source of interference at the \ac{TN} \ac{UE} will be from other \acp{BS} located in the same cluster with the \ac{TN} \ac{UE} and the serving \ac{BS}, as well as the \ac{LEO} satellite that provides coverage over the \ac{TN} cluster. 

We assume that \acp{BS} are distributed according to a \ac{BPP} on a circle with radius $\rTN$, and that the victim \ac{TN} \ac{UE}, without loss of generality, is located on point $(x_0,0,0)$. Now let us define a circle centered at \ac{UE}'s location with radius equal to the distance between the \ac{UE} and any \ac{BS} in the \ac{TN} cluster, denoted by $R_n$. The area of this circle is given by
\begin{align}
    \label{eq:A_user}
&{\cal A}(r_n)=\\\nonumber
&\hspace{-1 mm}\left\{
\begin{array}{ll}
{\pi r_n^2}, & 0\leq r_n\leq \rTN-x_0\\
 r_n^2\left(\theta-\frac{1}{2}\sin(2\theta)\right)&\\
 +\rTN^2\left(\phi-\frac{1}{2}\sin(2\phi)\right), & \rTN-x_0<r_n < \rTN+x_0,\\
0,&\text{otherwise}.
\end{array} \right.
\end{align}
where $\theta=\arccos\left(\frac{r_n^2+x_0^2-\rTN^2}{2x_0 r_n}\right)$ and $\phi=\arccos\left(\frac{-\rn^2+x_0^2+\rTN^2}{2x_0\rTN}\right)$. Thus, for \ac{BPP}-distributed \acp{BS}, the \ac{CDF} of the distance between the \ac{TN} \ac{UE} and any arbitrary \ac{BS} is $F_{\Rn}\left(\rn\right)=\frac{{\cal A}(\rn)}{\pi\rTN^2}$.
Due to the channel assignment by which the serving \ac{BS} is the nearest one among all the $\Nc$ i.i.d. \acp{BS}, the CDF of $R_0$ for a cluster of \ac{BPP} distributed gNBs can be written as \eqref{eq:CDF_R0}.

When conditioned on the serving distance such that $\Rs=\rs$, the CDF of $R_n|_{R_0=r_0}$ is obtained by conditioning $\Rn$ on $\Rs$ as follows:
\setcounter{equation}{3}
\begin{align}
\label{eq:Rn_R0_CDF}
F_{R_n|\Rs}(r_n|\rs)&\triangleq\Prob\left(R_n<r_n|R_0=r_0\right)\\\nonumber
&=\frac{\Prob\left(r_0 \leq R_n \leq r_n\right)}{\Prob(R_n>r_0)} = \frac{F_{R_n}(r_n)-F_{R_n}(r_0)}{1-F_{R_n}(r_0)}.
\end{align}
Taking the derivative w.r.t. $r_n$ gives the probability density function (PDF) of $R_n$ as
\begin{align}
\label{eq:Rn_R0_PDF}
f_{R_n|\Rs}(r_n|\rs)= \frac{f_{R_n}(r_n)}{1-F_{R}(r_0)}.
\end{align}
The interference power received from a satellite transmitting towards the ground users, i.e., $\Intn$ in \eqref{eq:SINR}, can be obtained as $\IntnI = \Pntn \Gntn \Hntn \Rntnsat^{-\antn}$,  where $\Rntnsat$ is the distance between the LEO satellite and the ground users, and $\antn$ is a path loss exponent for NTN link. 




\subsection{Co-existence Case~II: NTN UL as the aggressor and TN DL as the victim}

Co-existence Case~II corresponds to Scenario~5 given in \cite{TR38.863} in which \ac{TN} DL and \ac{NTN} UL send the data over the same frequency band (c.f. Fig.\ref{fig:sys_scenario5}). The aggressors in this case are the \ac{NTN} UEs and the victim, similarly to Case~I, is the \ac{TN} \ac{UE}, which receives the useful signal from its nearest TN \ac{BS}. Thus, the source of interference at the \ac{TN} \ac{UE} will be from other \ac{NTN} UEs, located at the edge of the \ac{TN} \ac{UE} cluster, and other \acp{BS} sharing the same cluster with the serving \ac{BS}.  Therefore, the serving distance distribution and TN interfering distances distribution will be similar to Case~I as given in \eqref{eq:CDF_R0} and \eqref{eq:Rn_R0_CDF}, respectively. The NTN interference power can be written as $\IntnII = \sum_{n=1}^{\Nu} \Pntn \Gntn \Hntn \Rntn^{-\atn}$.

\begin{table*}
\renewcommand \caption [2][]{}
\centering
\begin{minipage}{0.9\textwidth}
\begin{align}
    \label{eq:CDF_Rn_ntn}
     &F_{\Rntn}(\mathcal{r}_n)=\frac{1}{\pi(\rntn^2-(\rTN+\riso)^2)}\times\\\nonumber
&\left\{
\begin{array}{ll}
0, & 0\leq\Rntn\leq w-x_0,\\
&\\
F_{w,0}\left(\frac{w^2-{\mathcal{r}_n}^2+x_0^2}{2x_0}\right)-F_{{\mathcal{r}_n},x_0}\left(\frac{w^2-{\mathcal{r}_n}^2+x_0^2}{2x_0}\right)+F_{{\mathcal{r}_n},x_0}\left(x_0+{\mathcal{r}_n}\right)-F_{w,0}\left(w\right), &w-x_0<{\mathcal{r}_n} \leq \rntn-x_0,\\
&\\
F_{w,0}\left(\frac{w^2-{\mathcal{r}_n}^2+x_0^2}{2x_0}\right)-F_{{\mathcal{r}_n},x_0}\left(\frac{w^2-{\mathcal{r}_n}^2+x_0^2}{2x_0}\right)+F_{{\mathcal{r}_n},x_0}\left(\frac{\rntn^2-{\mathcal{r}_n}^2+x_0^2}{2x_0}\right)+&\rntn-x_0<{\mathcal{r}_n} \leq w+x_0,\\
F_{\rntn,0}\left(\rntn\right)-F_{w,0}\left(w\right)-F_{\rntn,0}\left(\frac{\rntn^2-{\mathcal{r}_n}^2+x_0^2}{2x_0}\right), &\\
&\\
\pi\left({\mathcal{r}_n}^2-w^2\right)-F_{\rntn,0}\left(\rntn\right)-F_{{\mathcal{r}_n},x_0}\left(x_0+{\mathcal{r}_n}\right)+ &w+x_0\leq{\mathcal{r}_n}\leq \rntn+x_0,\\
F_{{\mathcal{r}_n},x_0}\left(\frac{\rntn^2-{\mathcal{r}_n}^2+x_0^2}{2x_0}\right)+F_{\rntn,0}\left(\frac{\rntn^2-{\mathcal{r}_n}^2+x_0^2}{2x_0}\right), &\\
&\\
1, & \text{otherwise}.
\end{array} \right.
\end{align}
\end{minipage}
\end{table*}

To obtain the distribution of the distances between the TN UE and NTN UEs, we assume that NTN UEs are distributed according to a \ac{BPP} on the annulus between two cocentric circles with radii $\rTN+\riso$ and $\rntn$ where $\rTN+\riso<\rntn$ and $\riso\geq 0$, as depicted in Fig.~\ref{fig:Integral bounderies}. Obviously, in this case, the interfering distances are independent of the serving distance. Thus, the distribution is not conditional on $\Rs$. The distances from the TN UE to any NTN UE is denoted by $\Rntn$ to distinguish it from $R_n$, the distances between TN \acp{BS} and TN UE. The CDF of $\Rntn(\mathcal{r}_n)$ is obtained as the ratio of the shaded region area shown in Fig~\ref{fig:Integral bounderies} to the total surface area of the annulus where NTN UEs are distributed as a \ac{BPP}. The shaded area is the area obtained from the intersection of the disc centered at TN UE with radius $\Rntn$ and the outer annulus.

The CDF $F_{\Rntn}(\mathcal{r}_n)=\frac{{\cal A}(\mathcal{r}_n)}{\pi(\rntn^2-(\rTN+\riso)^2)}$ is a piece-wise function given in \eqref{eq:CDF_Rn_ntn}, where ${\cal A}(\mathcal{r}_n)$ is the area of shaded region, $w=\rTN+\riso$, $F_{a,b}\left(y\right)=(y-b)\sqrt{a^2-(b-y)^2}+a^2\tan^{-1}\left(\frac{y-b}{\sqrt{a^2-(b-y)^2}}\right)$. As can be seen from \eqref{eq:CDF_Rn_ntn}, the distribution of interfering distances depends on the location of TN user and the geometry of the network in terms of $\rTN$, $\riso$, and $\rntn$. The PDF of $\Rntn$ can be obtained by taking the derivative w.r.t. $\mathcal{r}_n$.

\section{Exact Performance Analysis}
\label{sec:cov}
In this section, we focus on the performance analysis of the integrated network for the scenarios described in Section~\ref{sec:sysmod}, in terms of coverage probability and data rate of a \ac{TN} \ac{UE} which is arbitrarily located on Earth. We utilize stochastic geometry in order to formulate coverage probability and rate as a function of the network's parameters and the propagation characteristics of the channels. One main component of our analytical derivations is the Laplace function of interference which will be presented throughout this section for both cases.

\subsection{ Performance Analysis for Coexistence Case~I: NTN DL as the aggressor and TN DL as the victim}
We define the coverage probability as the probability that the \ac{SINR} at the TN UE's receiver is greater than a minimum required \ac{SINR}, denoted by $T$. Thus, we have
\begin{align}
\label{eq:coverage probability}
\Pc\left(T\right) \triangleq \Prob\left(\SINRTN > \hspace{-1 mm}T\right).
\end{align}

Based on the above given definition, in the following theorem, we derive an analytical expression for the probability of coverage for co-existence case~I over a Nakagami-$m$ fading serving channel.

\begin{theorem}
    \label{theorem:coverage senario3}
    The coverage probability for a \ac{TN} user in Co-existence Case~I is
    \begin{align}
\label{eq:Coverage_probability_scenario3}
&\Pc\left(T\right)=\int_{0}^{\rTN+x_0} f_{\Rs}\left(\rs\right)\Bigg[e^{-s \sigma^2}\\\nonumber
&\sum_{k=0}^{m-1}\frac{\sum_{l=0}^{k}\hspace{-1 mm}\binom{k}{l}\hspace{-1 mm}\left(s \sigma^2\right)^l \left(-s\right)^{k-l}\hspace{-1 mm}\frac{\partial^{k-l}}{\partial s^{k-l}}\left(\mathcal{L}_{\Itn}(s)\mathcal{L}_{\IntnI}(s)\right)}{k!}\Bigg]_{s=\Sc} \hspace{-5 mm}d\rs,
\end{align}
where $\Sc=\frac{mT\rs^{\atn}}{\Ptn \Gtn}$, $\mathcal{L}_{\Itn}\left(s\right)$ and $\mathcal{L}_{\Itn}\left(s\right)$ are the Laplace transform of cumulative interference power $\Itn$ and $\IntnI$ which are expressed in Lemma~\ref{lem:Laplace scenario 3} and Lemma~\ref{lem:Laplace ntn scenario 3}, respectively.
\end{theorem}
\begin{proof}
 See Appendix~\ref{Appen:proof coverage senario3}.
\end{proof}

\begin{lem}
\label{lem:Laplace scenario 3}
Laplace function of cumulative interference power $\Itn$ under general fading channels is
\begin{align}
\label{eq:Laplace_TN}
&\mathcal{L}_{\Itn}(s)\triangleq\mathbb{E}_{\Itn}\left[e^{-s\Itn}\right]\\\nonumber
&= \left( \int_{r_0}^{\rTN+x_0}\hspace{-6 mm}\mathcal{L}_{\Htn }\left(s\Ptn \Gtn r_n^{-\atn}\right)f_{R_n|\Rs}(\rn|\rs)\,dr_n\right)^{\Nc}\hspace{-3 mm}.
\end{align}
\end{lem}

\begin{proof}
See Appendix~\ref{Appen:proof Laplace scenario3}
\end{proof}

\begin{lem}
\label{lem:Laplace ntn scenario 3}
Laplace function of cumulative interference power $\Itn$ under general fading channels is
\begin{align}
\label{eq:Laplace_NTN_CaseI}
\nonumber
&\mathcal{L}_{\IntnI}(s)\triangleq\mathbb{E}_{\Rntnsat, \IntnI}\left[e^{-s\IntnI}\right]=\\\nonumber
&\int_{a}^{\rmax}\hspace{-4 mm}\mathbb{E}_{\Hntn}\left[\exp{\left(-s\Pntn \Gntn \Hntn \rntnsat^{-\antn}\right)}\right]f_{\Rntnsat}(\mathcal{r}_n)\,d\mathcal{r}_n\\
&=\int_{a}^{\rmax}\mathcal{L}_{\Hntn}(s\Pntn\Gntn \rntnsat^{-\antn})f_{\Rntnsat}(\mathcal{r}_n)\,d\mathcal{r}_n,
\end{align}
where $f_{\Rntnsat}(\mathcal{r}_n)$ is the PDF of the distance between the satellite and a TN user which can be obtained from \cite[Lemma~1]{okati2020downlink}, and $\rmax=\sqrt{2\rEarth a+a^2}$ denotes the maximum possible distance between a satellite and a UE that is realized when the satellite is at the user's horizon. 
\end{lem}

\begin{proof}
The Laplace function is obtained by taking expectations over the two random variables, i.e., the distance between the LEO satellite and TN UE, $\Rntnsat$, and the channel gain $\Hntn$. 
\end{proof}

As a result, assuming some specific fading distributions, $\mathcal{L}_{\Hntn}(\cdot)$ can be specified at the point $s\Pntn\Gntn \rntnsat^{-\antn}$. In other words, the Laplace transform of NTN interference for any fading distribution can be calculated using Lemma~\ref{lem:Laplace ntn scenario 3}. For instance, when $\Hntn$ has a gamma distribution\footnote{The gain of the channel, which is the square of Nakagami random variable, follows a gamma distribution.}, $\mathcal{L}_{\Hntn}(s\Pntn\Gntn \rntnsat^{-\antn})=\frac{m^{m}}{(m+s\Pntn\Gntn \rntnsat^{-\antn})^{m}}$, where $m$ is the Nakagami-$m$ fading parameter. Assuming LEO satellite to be at the zenith of the user, the Laplace function in Lemma~\ref{lem:Laplace ntn scenario 3} will be further simplified to $\mathcal{L}_{\IntnI}(s)=\frac{m^{m}}{(m+s\Pntn\Gntn\, a^{-\antn})^{m}}$.

The average achievable data rate (in bit/s/Hz) is defined as the ergodic capacity for a fading communication link which is derived from Shannon-Hartley theorem and normalized to unit bandwidth.
The average achievable rate is defined as
\begin{align}
\label{eq:data rate def.}
\bar{C} \triangleq\mathbb{E}\left[\log_2\left(1+\SINRTN\right)\right].
\end{align}
In the following theorem, we derive the expression for the average rate of the victim \ac{TN} user over Nakagami-$m$ fading serving channel. Both \ac{TN} and \ac{NTN} interfering channels may follow any arbitrary distribution.
\begin{figure}
        
         \centering
         \includegraphics[trim = 5mm 0mm 5mm 0mm, clip,width=\textwidth]{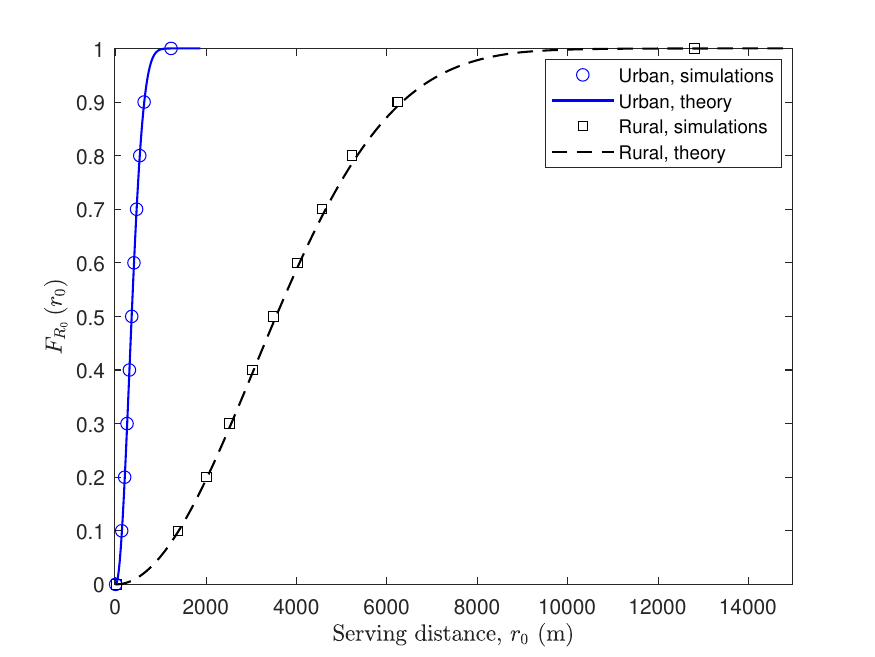}
         \caption{Verification of serving distance distribution given in \eqref{eq:CDF_R0} for urban (ISD = 0.75~km) and rural (ISD = 7.5~km) areas.}
         \label{fig:server distance}
\end{figure}
\begin{figure*}
         \begin{subfigure}[b]{0.48\textwidth}
         \centering
         \includegraphics[trim = 5mm 0mm 5mm 0mm, clip,width=\textwidth]{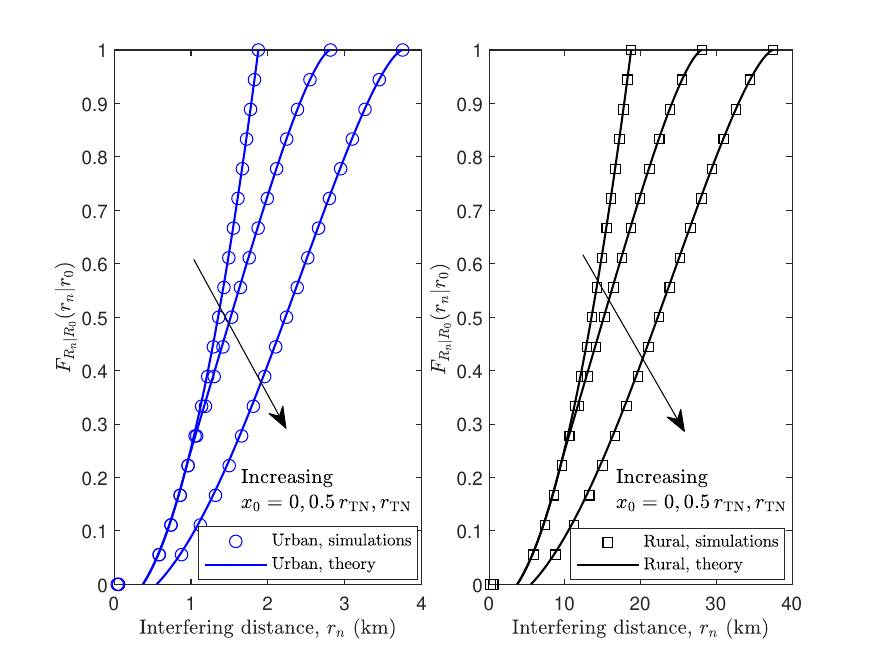}
         \caption{Case I.}
         \label{fig:Interfering distance scenario 3}
     \end{subfigure}
     \hfill
     \begin{subfigure}[b]{0.48\textwidth}
         \centering
         \includegraphics[trim = 5mm 0mm 5mm 0mm, clip,width=\textwidth]{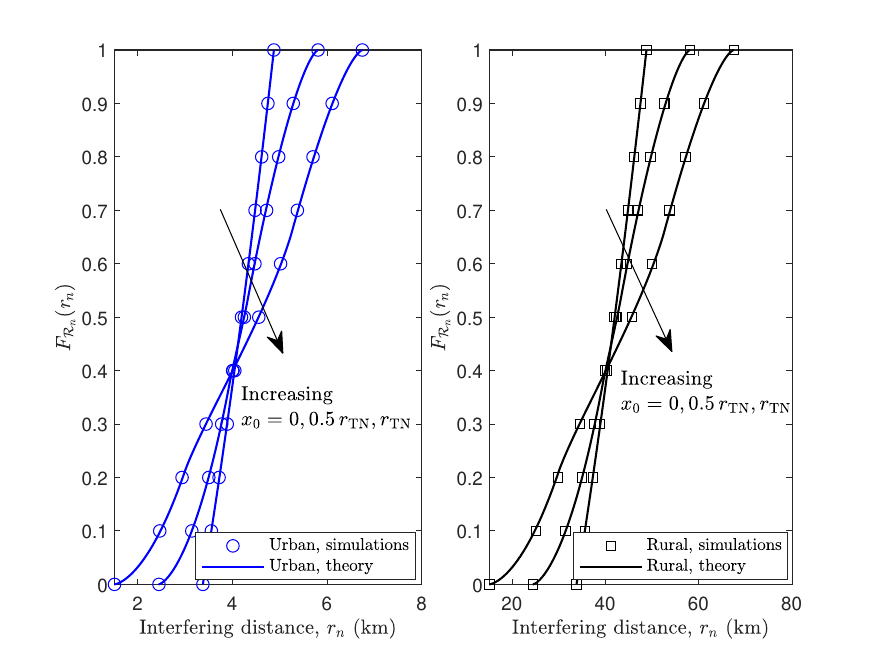}
         \caption{Case II.}
         \label{fig:interference distance scenario 5}
     \end{subfigure}
     \hfill 
    \caption{Verification of the interfering distance distribution, for cases~I (a), given as in \eqref{eq:Rn_R0_CDF}, and II (b) given as in \eqref{eq:Rn_R0_CDF}.}
   \label{fig:interfering dist}
\end{figure*}
\begin{theorem}
    \label{theorem:rate senario3}
    The average achievable rate for a \ac{TN} user in Co-existence Case~I is
    \begin{align}
\label{eq: rate senario3}
&\bar{C}=\int_{0}^{\rTN+x_0}\hspace{-2 mm}\int_0^{\infty}f_{\Rs}\left(\rs\right)\Bigg[e^{-{s\sigma^2 }}\\\nonumber
&\sum_{k=0}^{m-1}\frac{\sum_{l=0}^{k}\hspace{-1 mm}\binom{k}{l}\hspace{-1 mm}\left(s \sigma^2\right)^l \left(-s\right)^{k-l}\hspace{-1 mm}\frac{\partial^{k-l}}{\partial s^{k-l}}\left(\mathcal{L}_{\Itn}(s)\mathcal{L}_{\IntnI}(s)\right)}{k!}\Bigg]_{s=\sr} \hspace{-6 mm}dt d\rs,
\end{align}
where $\sr=\frac{m\left(2^t-1\right)\rs^{\atn}}{\Ptn \Gtn}$,  $\mathcal{L}_{\Itn}\left(s\right)$ and $\mathcal{L}_{\IntnI}$ are obtained from Lemmas~\ref{lem:Laplace scenario 3} and \ref{lem:Laplace ntn scenario 3}, respectively.
\end{theorem}
\begin{proof}
 See Appendix~\ref{Appen:proof rate scenario 3}.
\end{proof}

\subsection{ Performance Analysis for Co-existence Case~II: NTN UL as the aggressor and TN DL as the victim}
In the following theorem, we analytically formulate the probability of coverage for co-existence case~II, utilizing the definition provided earlier in \eqref{eq:coverage probability}. Similar to Theorem~\ref{theorem:coverage senario3}, the serving channel follows a Nakagami-$m$ distribution while interfering channels may follow arbitrary distributions.

\begin{theorem}
    \label{theorem:coverage senario5}
    The coverage probability for a \ac{TN} user in Co-existence Case~II is
    \begin{align}
\label{eq:Coverage_probability_scenario5}
&\Pc\left(T\right)=\int_{0}^{\rTN+x_0} f_{\Rs}\left(\rs\right)\Bigg[e^{-s \sigma^2}\\\nonumber
&\sum_{k=0}^{m-1}\frac{\sum_{l=0}^{k}\hspace{-1 mm}\binom{k}{l}\hspace{-1 mm}\left(s \sigma^2\right)^l \hspace{-1 mm}\left(-s\right)^{k-l}\hspace{-1 mm}\frac{\partial^{k-l}}{\partial s^{k-l}}\left(\mathcal{L}_{\Itn}(s)\mathcal{L}_{\IntnII}(s)\right)}{k!}\Bigg]_{s=\Sc} \hspace{-6 mm}d\rs,
\end{align}
where $\Sc=\frac{mT\rs^{\atn}}{\Ptn \Gtn}$, and $\mathcal{L}_{\IntnII}\left(s\right)$ is the Laplace transform of cumulative interference power $\IntnII$ from all NTN UEs located on the outer annulus shown in Fig.~\ref{fig:Integral bounderies} that is expressed in Lemma~\ref{lem:Laplace scenario 5}. Obviously, $\mathcal{L}_{\Itn}\left(s\right)$ remains the same as for Case~I and is obtained from Lemma~\ref{lem:Laplace scenario 3}.
\end{theorem}
\begin{proof}
The expression can be obtained using similar steps as shown in the proof of Theorem~\ref{theorem:coverage senario3} given in Appendix~\ref{Appen:proof coverage senario3} by only substituting $\IntnII\to\IntnI$. 
\end{proof}

\begin{lem}
\label{lem:Laplace scenario 5}
Laplace function of cumulative interference power $\IntnII$ under general fading channels is
\begin{align}
\label{eq:Laplace_NTN}
&\mathcal{L}_{\IntnII}(s)\triangleq\mathbb{E}_{\IntnII}\left[e^{-s\IntnII}\right]\\\nonumber
&= \left( \int_{w-x_0}^{\rntn+x_0}\mathcal{L}_{\Hntn }\left(s\Pntn \Gntn \mathcal{r}_n^{-\atn}\right)f_{\Rntn}(\mathcal{r}_n)\,d\mathcal{r}_n\right)^{\Nu}.
\end{align}
\end{lem}

\begin{proof}
See Appendix~\ref{Appen:proof Laplace scenario 5}
\end{proof}

In the following theorem, we formulate the expression for the average rate of the victim \ac{TN} user over Nakagami-$m$ fading serving channel for co-existence Case~II.
\begin{theorem}
    \label{theorem:rate senario5}
    The average achievable rate for a \ac{TN} user in Co-existence Case~II is
    \begin{align}
\label{eq: rate senario5}
&\bar{C}=\int_{0}^{\rTN+x_0}\hspace{-2 mm}\int_0^{\infty}f_{\Rs}\left(\rs\right)\Bigg[e^{-{s\sigma^2 }}\\\nonumber
&\sum_{k=0}^{m-1}\frac{\sum_{l=0}^{k}\hspace{-1 mm}\binom{k}{l}\hspace{-1 mm}\left(s\sigma^2\right)^l \hspace{-1 mm}\left(-s\right)^{k-l}\hspace{-1 mm}\frac{\partial^{k-l}}{\partial s^{k-l}}\hspace{-1 mm}\left(\mathcal{L}_{\Itn}(s)\mathcal{L}_{\IntnII}(s)\right)}{k!}\hspace{-1 mm}\Bigg]_{s=\sr} \hspace{-6 mm}dt d\rs,
\end{align}
where $\sr=\frac{m\left(2^t-1\right)\rs^{\atn}}{\Ptn \Gtn}$. 
$\mathcal{L}_{\Itn}\left(s\right)$ and $\mathcal{L}_{\IntnII}\left(s\right)$ 
are obtained from Lemmas~\ref{lem:Laplace scenario 3} and \ref{lem:Laplace scenario 5}, respectively.
\end{theorem}
\begin{proof}
The proof follows similar procedure as shown in the proof of Theorem~\ref{theorem:rate senario3} presented in Appendix~\ref{Appen:proof rate scenario 3}. The sole difference is the origin of NTN cumulative interference power, which, in this case, arises from NTN UL.
\end{proof}
\begin{figure*}
         \begin{subfigure}[b]{0.48\textwidth}
         \centering
         \includegraphics[trim = 5mm 0mm 5mm 0mm, clip,width=\textwidth]{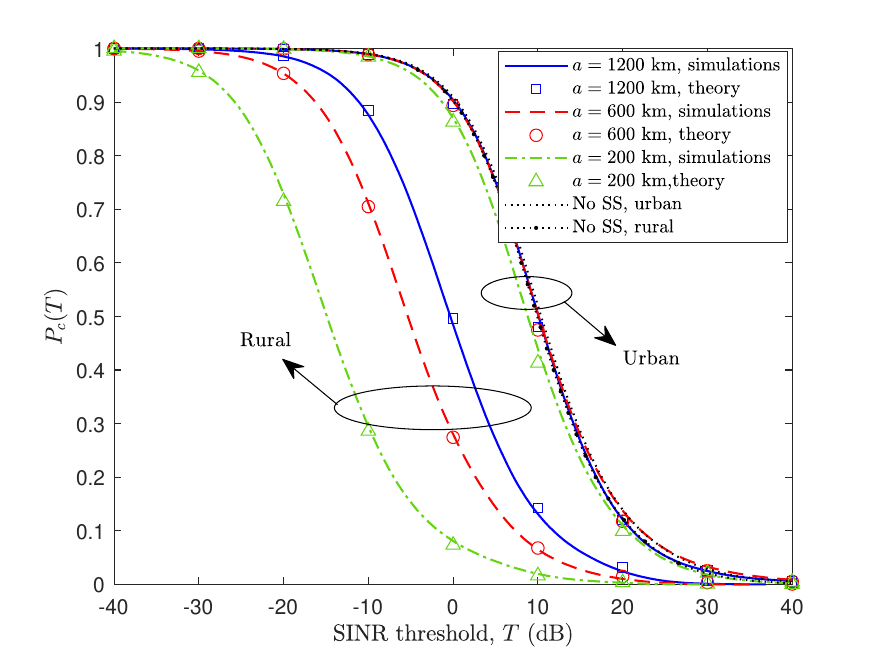}
         \caption{100\% load.}
         \label{fig:100load}
     \end{subfigure}
     \hfill
     \begin{subfigure}[b]{0.48\textwidth}
         \centering
         \includegraphics[trim = 5mm 0mm 5mm 0mm, clip,width=\textwidth]{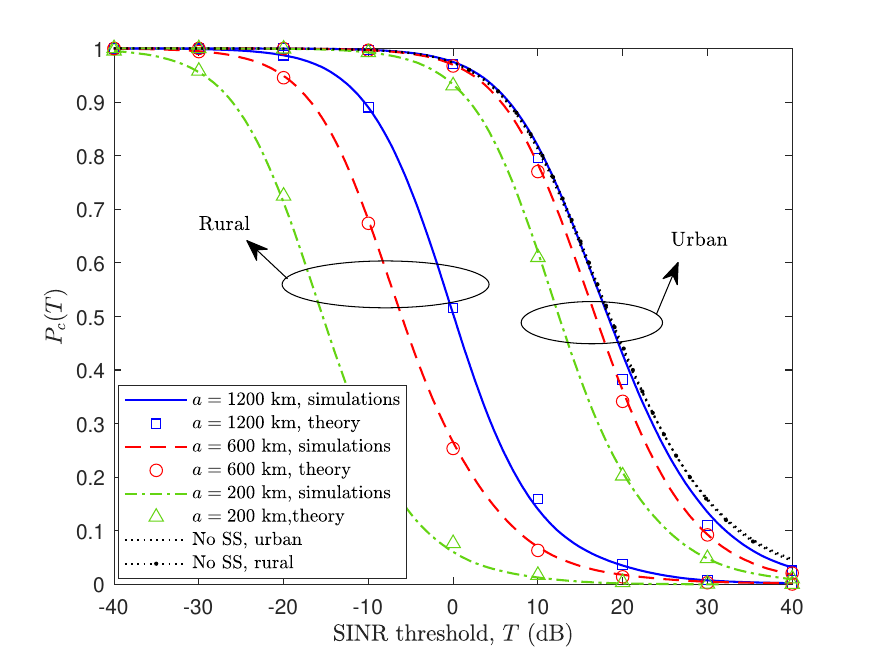}
         \caption{25\% load.}
         \label{fig:25load}
     \end{subfigure}
     \hfill 
    \caption{Coverage probability of Case~I for 100\% load (a), and 25\% load (b), for different LEO satellite altitudes. The simulations, shown by lines, verify the theoretical results, depicted by markers, given in Theorem~\ref{theorem:coverage senario3}.}
   \label{fig:Coverage_Scenario3}
\end{figure*}
\begin{figure*}
         \begin{subfigure}[b]{0.48\textwidth}
         \centering
         \includegraphics[trim = 5mm 0mm 5mm 0mm, clip,width=\textwidth]{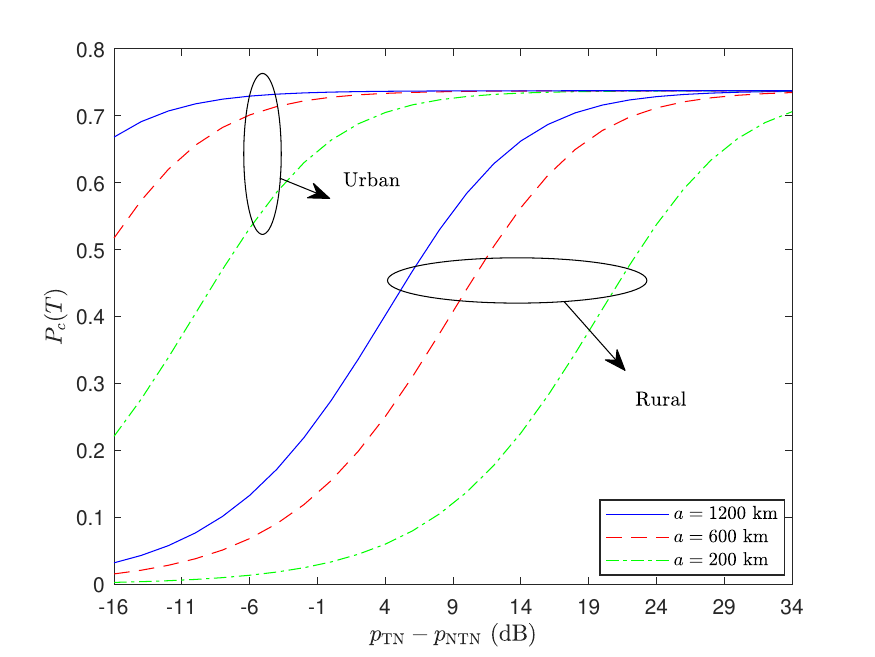}
         \caption{100\% load.}
         \label{fig:100load_pt}
     \end{subfigure}
     \hfill
     \begin{subfigure}[b]{0.48\textwidth}
         \centering
         \includegraphics[trim = 5mm 0mm 5mm 0mm, clip,width=\textwidth]{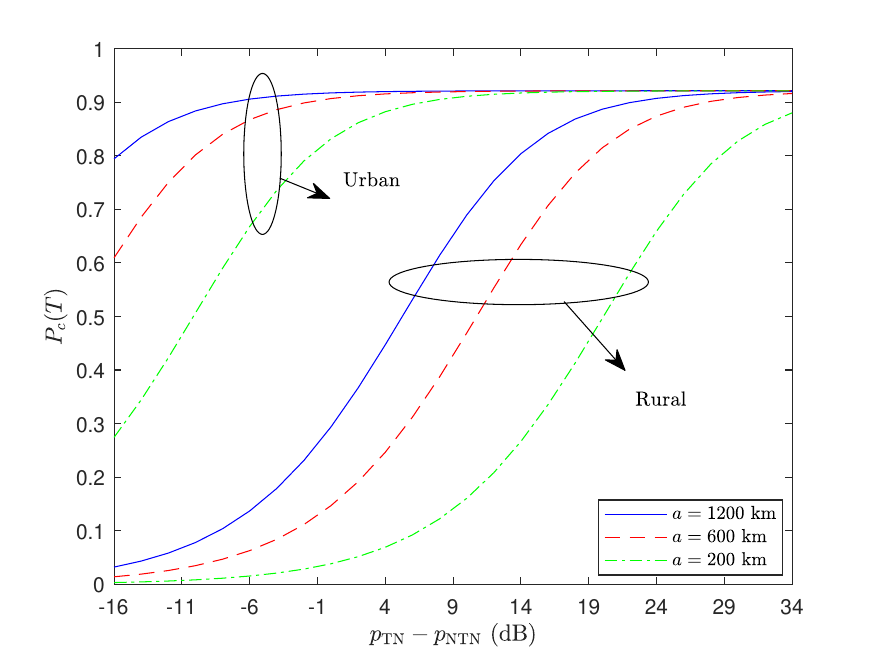}
         \caption{25\% load.}
         \label{fig:25load_pt}
     \end{subfigure}
     \hfill 
    \caption{Effect of transmit power on coverage probability of Case~I for 100\% load (a), and 25\% load (b).}
   \label{fig:Coverage_Scenario3VSpt}
\end{figure*}
\section{Numerical Results}
\label{sec:Numerical Results}
{\color{black}In this section, the analytical derivations on the performance of the integrated network for each case is verified through Monte Carlo simulations. }
The \ac{TN} \acp{BS}, satellites, and UEs antenna models are chosen according to the sectorized antenna patterns given in \cite{TR38.863} for simulations. However, in order to maintain tractability for the theoretical results, we approximate such model by assuming fixed antenna gains for the main and side lobes of the antennas, similar to \cite{bole2013radar,renzo2015Stochastic,singh2015Tractable}. In fact, the main beam transmits with the maximum antenna gain given in \cite{TR38.863}, while side lobes are assumed to have 13~dB lower gain, similar to \cite{bole2013radar,TR138921}.

We consider the path-loss exponents to be $\atn=3$ and $\antn=2$ for \ac{TN} and \ac{NTN} propagation channels, respectively, according to \cite[Table~3.2]{rappaport2010wireless}. We adopt a practical fading model, i.e., Rician, for all simulation results with parameter $K$, where $K=0$ and 200 for TN and NTN channels, respectively, to account for higher LoS in \ac{NTN}.  Since the theoretical derivations are developed assuming a Nakagami-$m$ fading for the serving link to preserve the analytical tractability, a Nakagami-$m$ distribution is used for their numerical calculations by setting its parameter to $m=\frac{k^2+2K+1}{2K+1}$, by matching the first and second moments of Rician and Nakagami PDFs. {\color{black}The fair match between theoretical results and simulations confirms the high accuracy of such approximation.}


The satellite transmit power, cell radius, maximum antenna gain are $\pt=46$~dBm, 25~km, and 30~dBi, respectively. The main lobe gain of the \acp{BS} antennas is assumed to be $17$~dBi. The \ac{BS} transmit power is assumed to be $\pt=46$~dBm. The altitude is set to $\{200,600,1200\}$~km, unless stated otherwise. Inter-site distances for \ac{TN} network are assumed to be 0.75~km and 7.5~km for urban and rural areas, respectively. The number of \acp{BS} on each \ac{TN} cluster is assumed to be 19.
The operating frequency and bandwidth are set to be 2~GHz and 20~MHz, respectively. 

Figure~\ref{fig:server distance} verifies the derivation of serving distance distribution given in \eqref{eq:CDF_R0} for urban and rural areas. It is worth noting that the serving distance is the same for both Cases~I and II. 
Figures~\ref{fig:interfering dist}(a) and \ref{fig:interfering dist}(b) verify the CDF of interference from TN \acp{BS} and NTN UEs derived in in \eqref{eq:Rn_R0_CDF} and \eqref{eq:CDF_Rn_ntn}, respectively, in rural and urban areas. The results are presented for a TN user located at the center, middle, and the edge of the TN cluster, i.e., $x_0=\{0, 0.5\,\rTN, \rTN\}$. As can be seen in the figures~\ref{fig:server distance} and \ref{fig:interfering dist}, there is a perfect match between the analytical results, shown by lines, and the simulations, shown by markers, which verifies the exactness of the derivations. 

\begin{figure*}
         \begin{subfigure}[b]{0.48\textwidth}
         \centering
         \includegraphics[trim = 5mm 0mm 5mm 0mm, clip,width=\textwidth]{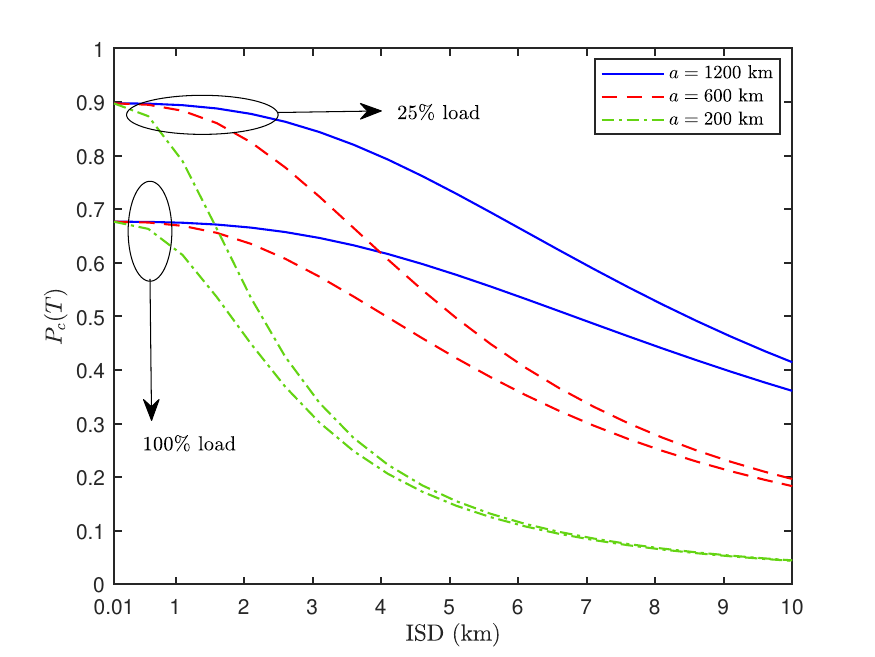}
         \caption{Coverage probability}
         \label{fig:coverage_ISD_Sc3}
     \end{subfigure}
     \hfill
     \begin{subfigure}[b]{0.48\textwidth}
         \centering
         \includegraphics[trim = 5mm 0mm 5mm 0mm, clip,width=\textwidth]{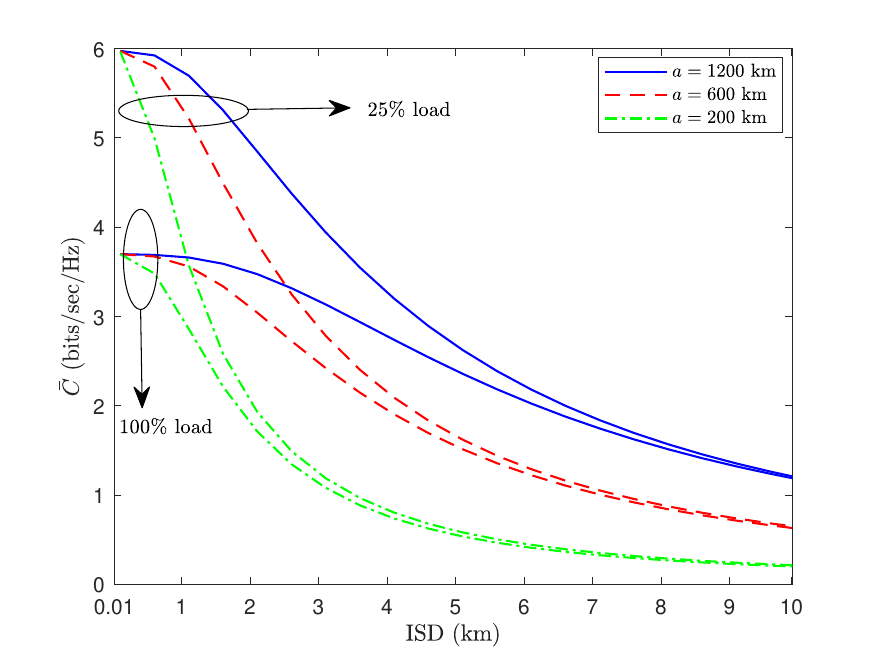}
         \caption{Data rate}
         \label{fig:Rate_ISD_Sc3}
     \end{subfigure}
     \hfill 
    \caption{Effect of ISD on Coverage probability (a) and data rate (b) of Case~I for different LEO satellite altitudes and TN load levels.}
   \label{fig:Coverage_Rate_ISD_SC3}
\end{figure*}
\begin{figure*}
         \begin{subfigure}[b]{0.48\textwidth}
         \centering
         \includegraphics[trim = 5mm 0mm 5mm 0mm, clip,width=\textwidth]{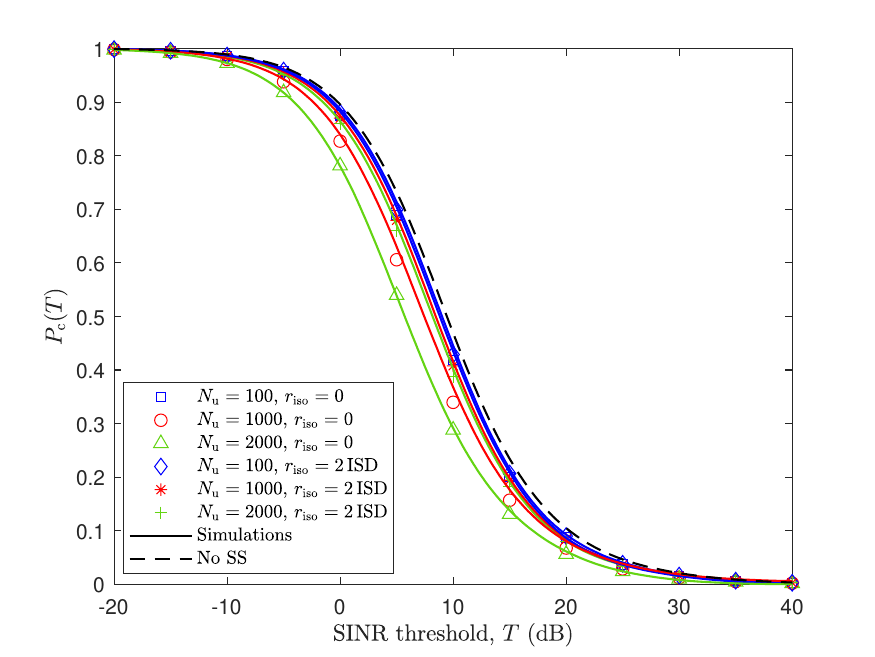}
         \caption{100\% load.}
         \label{fig:100load_sc5}
     \end{subfigure}
     \hfill
     \begin{subfigure}[b]{0.48\textwidth}
         \centering
         \includegraphics[trim = 5mm 0mm 5mm 0mm, clip,width=\textwidth]{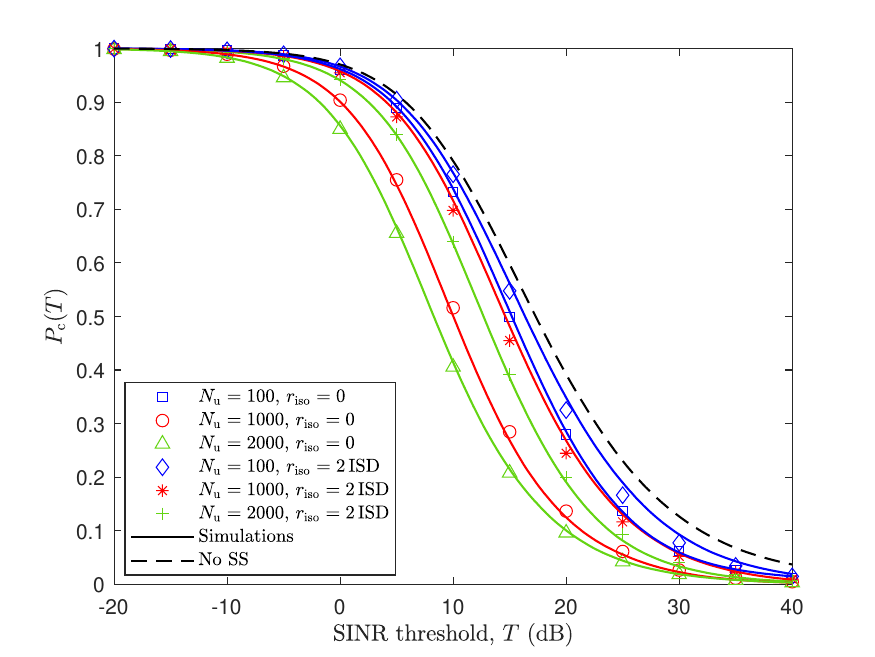}
         \caption{25\% load.}
         \label{fig:25load_sc5}
     \end{subfigure}
     \hfill 
    \caption{Coverage probability of Case~II for 100\% load (a), and 25\% load (b) for different number of NTN UEs and isolation distances. The simulations, shown by lines, verify the theoretical results, depicted by markers, given in Theorem~\ref{theorem:coverage senario5}.}
   \label{fig:Coverage_Scenario5}
\end{figure*}
\begin{figure}   
         \centering
         \includegraphics[trim = 5mm 0mm 5mm 0mm, clip,width=\textwidth]{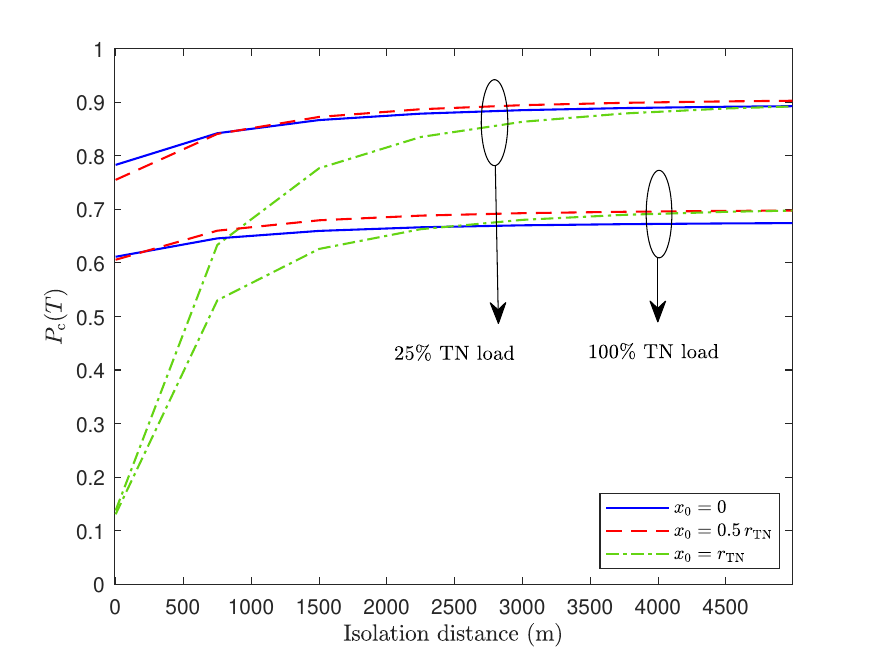}
         \caption{Effect of isolation distance, $\riso$, on coverage probability of Case~II for different TN UE's locations.}
         \label{fig:Coverage_Rate_riso_SC5}
\end{figure}

Figures~\ref{fig:Coverage_Scenario3}(a) and \ref{fig:Coverage_Scenario3}(b) verify the derivations in Theorem~\ref{theorem:coverage senario3} for $100\%$ and $25\%$ load from the TN network. The load level represents what percentage of BSs are active within the TN cluster. The black dotted lines represent the case of having no spectrum sharing between TN and NTN, i.e., the whole band is allocated only to TN. We can observe that the performance degradation due to sharing the spectrum is less significant in urban areas. The reason is that \acp{BS} are located in relatively closer proximity, i.e., smaller ISD, in urban areas, such that the interference received from TN \acp{BS} is much larger than the interference received from the LEO satellite, especially for $100\%$ load. On the other hand, for rural network, as the \acp{BS} are much farther from each other, NTN interference becomes more comparable to TN interference. Thus, there is more performance degradation for rural areas when sharing the spectrum between TN and NTN. Obviously, the impact of LEO satellite altitude is also more significant in rural regions.

Figures~\ref{fig:Coverage_Scenario3VSpt}(a) and \ref{fig:Coverage_Scenario3VSpt}(b) show the effect of power allocation between TN and NTN on coverage probability. The transmit power of LEO satellite is assumed to be fixed and equal to $\Pntn=46$~dBm. By increasing the TN power, the coverage probability saturates to a fixed value for which the NTN has no effect on its performance. The saturation value depends on the number of interfering TN BSs. As can be seen from the figure, a higher TN power is required for larger ISD distances, i.e., rural regions, as well as for lower satellite altitudes to maximize the coverage probability. 

The effect of ISD for different satellites altitudes and TN load levels on coverage probability and data rate are illustrated in Figs.~\ref{fig:Coverage_Rate_ISD_SC3}(a) and \ref{fig:Coverage_Rate_ISD_SC3}(b), respectively. Lower ISD provides better coverage and rate due to more strength of the serving signal. The performance drops with increasing the ISD. The decrease is more considerable for lower LEO altitudes due to higher interference power.

For Case~II, the transmit power of NTN UEs is assumed to be 200~mW with antenna gain of 1~dBi as specified in \cite{TR38.863}. Figures~\ref{fig:Coverage_Scenario5}(a) and \ref{fig:Coverage_Scenario5}(b) verify the derivations in Theorem~\ref{theorem:coverage senario5} for $100\%$ and $25\%$ load from the TN cell and varying the number of NTN UEs and isolation distances. The results show that the interference from TN DL is more significant compared to NTN UL, especially for higher loads. The effect of spectrum sharing is rather minimal for $100\%$ load from TN network, especially when isolation distance is greater than $2\isd$.

The effect of isolation distance between the TN and NTN networks is shown in Fig.~\ref{fig:Coverage_Rate_riso_SC5}, for $a=200$~km. As can be seen in the figure, increasing the isolation distance has no significant impact on the performance of users located at the center or middle of the TN cell. For edge users, a larger isolation distance is needed to achieve the highest possible performance. The minimum required isolation distance also decreases with increasing the TN load level since the NTN interference becomes less comparable to TN interference.

\section{Conclusions}
\label{sec:Conclusion}
In this paper, we studied the performance of integrated TN-NTN network when TN and NTN operate on the same frequency channel in S-band. Two co-existence cases were considered in which the TN users receive interference from either NTN DL or UL. Using tools from stochastic geometry, we obtained exact tractable expressions on the coverage probability and average data rate for each case in terms of TN and NTN networks' parameters.
The analysis and results provided herein will pave the way for optimal dynamic spectrum sharing between TN and NTN. The study showed that, depending on the network parameters, one co-existence case may outperform the other in terms of the probability of coverage and average achievable rate. Thus, the network may dynamically switch between the two cases to optimize the performance. It was also shown that there is a minimum isolation distance between TN and NTN which can minimize the effect of interference caused by spectrum sharing.

\appendix
\subsection{Proof of Theorem~\ref{theorem:coverage senario3}}
\label{Appen:proof coverage senario3}
Let us start derivation of Theorem~\ref{theorem:coverage senario3} using the definition of coverage probability given in \eqref{eq:coverage probability}:
\begin{align}
\label{eq:Coverage_probability_proof_part1}
\nonumber
&\Pc\left(T\right)=\Prob(\SINRTN > T )=\mathbb{E}_{R_0}\left[\Prob\left(\SINRTN>T |R_0=r_0\right)\right]\\\nonumber
&=\int_{0}^{u}\Prob\left(\SINRTN>T |R_0=r_0\right)f_{R_0}\left(r_0\right)dr_0\\\nonumber
&=\int_{0}^{u}\Prob\left(\frac{\Ptn \Gtn \Htn r_0^{-\atn}}{\Itn+\IntnI+\sigma^2}>T \right)f_{R_0}\left(r_0\right)dr_0\\\nonumber
&=\int_{0}^{u}\hspace{-2 mm}\mathbb{E}_{\Itn,\IntnI}\hspace{-1 mm}\left[\Prob\hspace{-1 mm}\left(\hspace{-1 mm}\Htn\hspace{-1 mm}>\hspace{-1 mm}\frac{T \left(\Itn+\IntnI+\sigma^2\right) r_0^{\atn}}{\Ptn \Gtn}\right)\hspace{-1 mm}\bigg|I_{\{\mathrm{TN,DL}\}}\hspace{-1 mm}>\hspace{-1 mm}0\right]\hspace{-1 mm}\\\nonumber
&\hspace{2 cm}\times f_{R_0}\left(r_0\right)dr_0\\\nonumber
&= \int_{0}^{u}\mathbb{E}_{\Itn,\IntnI}\left[1-F_{\Htn}\left(\frac{T \left(\Itn+\IntnI+\sigma^2\right) r_0^{\atn}}{\Ptn \Gtn}\right)\hspace{-1 mm}\right]\\\nonumber
&\hspace{2 cm}\times f_{R_0}\left(r_0\right)dr_0\\
\nonumber
&\stackrel{(a)}= \hspace{-1 mm}\int_{0}^{u}\hspace{-1 mm}\mathbb{E}_{\Itn,\IntnI}\hspace{-1 mm}\left[\frac{\Gamma\left(m,m\frac{T \left(\Itn+\IntnII+\sigma^2\right) r_0^{\atn}}{\Ptn \Gtn}\right)}{\Gamma\left(m\right)}\right]\hspace{-1 mm} f_{\Rs}\left(\rs\right)\, \hspace{-1 mm}d\rs\\\nonumber
&\stackrel{(b)}= \int_{0}^{u} f_{\Rs}\left(\rs\right)e^{-\frac{m\,T \sigma^2 r_0^{\atn}}{\Ptn \Gtn}}\mathbb{E}_{\Itn,\IntnI}\Bigg[e^{-m\frac{T \left(\Itn+\IntnI\right) r_0^{\atn}}{\Ptn \Gtn}}\\\nonumber
&\sum_{k=0}^{m-1}\frac{\sum_{l=0}^{k}\binom{k}{l}\left(m\frac{T \sigma^2 r_0^{\atn}}{\Ptn \Gtn}\right)^l \left(m\frac{T \left(\Itn+\IntnI\right) r_0^{\atn}}{\Ptn \Gtn}\right)^{k-l}}{k!}\Bigg] d\rs\\
&\stackrel{(c)}= \int_{0}^{u} f_{\Rs}\left(\rs\right)\Bigg[e^{-s \sigma^2}\\\nonumber
&\sum_{k=0}^{m-1}\frac{\sum_{l=0}^{k}\binom{k}{l}\left(s \sigma^2\right)^l \left(-s\right)^{k-l}\hspace{-1 mm}\frac{\partial^{k-l}}{\partial s^{k-l}}\mathcal{L}_{\Itn}(s)\mathcal{L}_{\IntnI}(s)}{k!}\Bigg]_{s=s_0}\hspace{-3 mm} d\rs,
\end{align}
where $u=\rTN+x_0$, $F_{\Htn}(\cdot)$ is the CDF of $\Htn$, (a) follows from the distribution of gamma random variable $\Htn$ (being the square of the Nakagami random variable), (b) is calculated by applying the definition of incomplete gamma function for integer values of $m$ to (a), and (c) follows from independence of $\Itn$ and $\Intn$, i.e., $\mathbb{E}_{\Intn,\Itn}\left[e^{-s(\Intn+\Itn)}\right]=\mathbb{E}_{\Intn}\left[e^{-s\Intn}\right] \mathbb{E}_{\Itn}\left[e^{-s\Itn}\right]$.

\label{Appen:proof of Lemma 2}

\subsection{Proof of Lemma~\ref{lem:Laplace scenario 3}}
\label{Appen:proof Laplace scenario3}
Using the definition of the Laplace transform yields
\begin{align}
\label{eq:proof laplace scenario 3 }
\nonumber
&\mathcal{L}_{\Itn}(s)\triangleq\mathbb{E}_{\Itn}\left[e^{-s\Itn}\right]\\\nonumber
&=\mathbb{E}_{ R_n, \Htn}\left[\exp{\left(-s\sum_{n=1}^{\Nc}\Ptn \Gtn \Htn R_n^{-\atn}\right)}\right]\\\nonumber
&=\mathbb{E}_{R_n, \Htn}\left[\prod_{n=1}^{\Nc}\exp{\left(-s\Ptn \Gtn \Htn R_n^{-\atn}\right)}\right]\\\nonumber
&\stackrel{(a)}= \mathbb{E}_{ R_n}\left[\prod_{n=1}^{\Nc}\mathbb{E}_{\Htn }\left[\exp{\left(-s\Ptn \Gtn \Htn R_n^{-\atn}\right)}\right]\right]\\\nonumber
&\stackrel{(b)}= \prod_{n=1}^{\Nc} \int_{r_0}^{u}\mathbb{E}_{\Htn }\left[\exp{\left(-s\Ptn \Gtn \Htn r_n^{-\atn} \right)}\right]\\\nonumber
&\hspace{2 cm}\times f_{R_n|\Rs}(\rn|\rs)\,dr_n\\
&\stackrel{(c)}= \left( \int_{r_0}^{u}\mathcal{L}_{\Htn }\left(s\Ptn \Gtn r_n^{-\atn}\right)f_{R_n|\Rs}(\rn|\rs)\,dr_n\right)^{\Nc},
\end{align}
where $u=\rTN+x_0$, (a) follows from the i.i.d.\ distribution of $\Htn $ and its further independence from $R_n$,  (b) is obtained using the interfering distance distribution from \eqref{eq:Rn_R0_PDF}, and (c) is derived thanks to the i.i.d.\ distribution of conditional PDF of $f_{R_n|\Rs}(\rn|\rs)$ and the substitution from the definition of Laplace function. 
\subsection{Proof of Theorem~\ref{theorem:rate senario3}}
\label{Appen:proof rate scenario 3}
Based on the definition of the average data rate given in \eqref{eq:data rate def.}, we have
\begin{align}
\nonumber
&\bar{C}=\mathbb{E}_{\Itn,\IntnI,\Htn,\Rs}\left[\log_2\left(1+\SINRTN\right)\right]\\
&=\hspace{0 mm}\int_{0}^{u}\hspace{0 mm}\mathbb{E}_{\Itn,\IntnI,\Htn}\hspace{0 mm}\left[\log_2\hspace{0 mm}\left(1+\SINRTN\right)\right]\hspace{0 mm}f_{\Rs}(\rs)\,d\rs\\\nonumber
&\stackrel{(a)}=\hspace{-2 mm}\int_{0}^{u}\hspace{-2 mm}\int_0^{\infty}\hspace{-1 mm}\mathbb{E}_{\Itn,\IntnI,\Htn}\left[\Prob\left(\log_2\hspace{-1 mm}\left(\hspace{0 mm}1+\SINRTN\hspace{0 mm}\right)\hspace{-1 mm}>\hspace{-1 mm} t\right)\right] f_{\Rs}(\rs)\,dt d\rs,
\end{align}
where $u=\rTN+x_0$, $\SINRTN=\frac{\Ptn \Gtn \Htn r_0^{-\atn}}{\Itn+\IntnI+\sigma^2}$, and (a) is obtained from the fact that for a positive random variable $X$, $\mathbb{E}\left[X\right]=\int_{t>0}\Prob\left(X>t\right)dt$. Thus, we have
\begin{align}
&\bar{C}\stackrel{(a)}=\\\nonumber
&\hspace{-1 mm}\int_{0}^{u}\hspace{-2 mm}\int_0^{\infty}\hspace{-2 mm}\mathbb{E}_{\Itn,\IntnI}\hspace{0 mm}\left[1-F_{\Htn}\hspace{-1 mm}\left(\frac{\rs^{\atn}\left(\Itn+\IntnI+\sigma^2\right) \left(2^t-1\right)}{\Ptn \Gtn}\hspace{-1 mm}\right)\hspace{-1 mm}\right] \\\nonumber
&\hspace{3 cm}\times f_{\Rs}(\rs)\, dt d\rs\\\nonumber
&\stackrel{(b)}=\int_{0}^{u}\hspace{-2 mm}\int_0^{\infty}\hspace{-2 mm}\mathbb{E}_{\Itn,\IntnI}\left[\frac{\Gamma\left(m,m\frac{\rs^{\atn}\left(\Itn+\IntnI+\sigma^2 \right)\left(2^t-1\right)}{\Ptn \Gtn}\right)}{\Gamma\left(m\right)}\right]\\\nonumber
&\hspace{3 cm}\times f_{\Rs}(\rs)\, dt d\rs\\\nonumber
&\stackrel{(c)}=\hspace{-1 mm}\int_{0}^{u}\hspace{-2 mm}\int_0^{\infty}\hspace{-3 mm}f_{\Rs}\left(\rs\right)e^{-\frac{m\, \sigma^2 \left(2^t-1\right)r_0^{\atn}}{\Ptn \Gtn}}\mathbb{E}_{\Itn,\IntnI}\Bigg[e^{-m\frac{ \left(\IntnI+\Itn\right) r_0^{\atn}}{\Ptn \Gtn}}\\\nonumber
&\sum_{k=0}^{m-1}\frac{\sum_{l=0}^{k}\hspace{-1 mm}\binom{k}{l}\hspace{-1 mm}\left(\hspace{-1 mm}m\frac{ \sigma^2\left(2^t-1\right) r_0^{\atn}}{\Ptn \Gtn}\right)^l \hspace{-2 mm} \left(\hspace{-1 mm}m\frac{\left(\IntnI+\Itn\right) r_0^{\atn}}{\Ptn \Gtn}\right)^{k-l}\hspace{-1 mm}}{k!}\hspace{-1 mm}\Bigg] dt d\rs,
\end{align}
where (a) follows from the product distribution of two independent random variables, (b) follows from the gamma distribution of serving channel gain $\Htn$, and (c) is calculated by applying the incomplete gamma function for integer values of $m$ to (b).

\subsection{Proof of Lemma~\ref{lem:Laplace scenario 5}}
\label{Appen:proof Laplace scenario 5}
Using the definition of the Laplace transform yields
\begin{align}
\nonumber
&\mathcal{L}_{\IntnII}(s)\triangleq\mathbb{E}_{\IntnII}\left[e^{-s\IntnII}\right]\\\nonumber
&=\mathbb{E}_{ \Rntn, \Hntn}\left[\exp{\left(-s\sum_{n=1}^{\Nu}\Pntn \Gntn \Hntn \Rntn^{-\atn}\right)}\right]\\\nonumber
&=\mathbb{E}_{\Rntn, \Hntn}\left[\prod_{n=1}^{\Nu}\exp{\left(-s\Pntn \Gntn \Hntn \Rntn^{-\atn}\right)}\right]\\\nonumber
&\stackrel{(a)}= \mathbb{E}_{ \Rntn}\left[\prod_{n=1}^{\Nu}\mathbb{E}_{\Hntn }\left[\exp{\left(-s\Pntn \Gntn \Hntn \Rntn^{-\atn}\right)}\right]\right]\\\nonumber
&\stackrel{(b)}= \prod_{n=1}^{\Nu} \int_{w-x_0}^{\rntn+x_0}\hspace{-3 mm}\mathbb{E}_{\Hntn }\hspace{-1 mm}\left[\exp{\left(-s\Pntn \Gntn \Hntn {\mathcal{r}_n}^{-\atn} \right)}\right]\\\nonumber
&\hspace{2.5 cm}\times f_{\Rntn}(\mathcal{r}_n)\,d\mathcal{r}_n\\
&\stackrel{(c)}= \left( \int_{w-x_0}^{\rntn+x_0}\hspace{-3 mm}\mathcal{L}_{\Hntn }\left(s\Pntn \Gntn \mathcal{r}_n^{-\atn}\right)f_{\Rntn}(\mathcal{r}_n)\,d\mathcal{r}_n\right)^{\Nu},
\end{align}
where $w=\rTN+\riso$, (a) follows from the i.i.d.\ distribution of $\Hntn $ and its further independence from $\Rntn$,  (b) is obtained using the derivative of interfering distance CDF from \eqref{eq:CDF_Rn_ntn}, (c) is derived thanks to the i.i.d.\ distribution of PDF of $f_{\Rntn}(\mathcal{r}_n)$ and the substitution from the definition of Laplace function. 

\bibliography{ref} 
\bibliographystyle{IEEEtran}



\end{document}